\documentclass[conference]{IEEEtran}

\usepackage[utf8]{inputenc} 
\usepackage[T1]{fontenc}
\usepackage{url}
\usepackage{ifthen}
\usepackage{cite}
\usepackage{amssymb,amsfonts}
\usepackage{algorithmic}
\usepackage{graphicx}
\usepackage{textcomp}
\usepackage{comment}
\usepackage{xcolor}
\usepackage{bbm}
\usepackage{bm}
\usepackage{balance}
\usepackage[cmex10]{amsmath}

\def\BibTeX{{\rm B\kern-.05em{\sc i\kern-.025em b}\kern-.08em
    T\kern-.1667em\lower.7ex\hbox{E}\kern-.125emX}}

\DeclareMathAlphabet\mathbfcal{OMS}{cmsy}{b}{n}

\newtheorem{thm}{Theorem}
\newtheorem{lem}{Lemma}

\begin{document}

\title{Space Upper Bounds for $\alpha$-Perfect Hashing}
\author{
  \IEEEauthorblockN{Ryan Song and Emre Telatar}
  \IEEEauthorblockA{School of Computer and Communication Sciences \\ 
        {\'E}cole Polytechnique F{\'e}d{\'e}rale de Lausanne (EPFL), Switzerland \\
        Emails: ryan.song@epfl.ch, emre.telatar@epfl.ch}
}
\maketitle

\begin{abstract}
     In the problem of minimal perfect hashing, we are given a size $k$ subset $\mathcal{A}$ of a universe of keys $[n] = \{1,2, \cdots, n\}$, for which we wish to construct a hash function $h: [n] \to [k]$ such that $h(\cdot)$ maps $\mathcal{A}$ to $[k]$ with no collisions, i.e., the restriction of $h(\cdot)$ to $\mathcal{A}$ is injective. In this paper, we extend the study of minimal perfect hashing to the approximate setting. For an $\alpha \in [0, 1]$, we say that a randomized hashing scheme is $\alpha$-perfect if for any input $\mathcal{A}$ of size $k$, it outputs a hash function which exhibits at most $(1-\alpha)k$ collisions on $\mathcal{A}$ in expectation. One important performance consideration for any hashing scheme is the space required to store the hash functions. For minimal perfect hashing, it is well known that approximately $k\log(e)$ bits, or $\log(e)$ bits per key, is required to store the hash function. In this paper, we propose schemes for constructing minimal $\alpha$-perfect hash functions and analyze their space requirements. We begin by presenting a simple base-line scheme which randomizes between perfect hashing and zero-bit random hashing. We then present a more sophisticated hashing scheme based on sampling which significantly improves upon the space requirement of the aforementioned strategy for all values of $\alpha$. 
\end{abstract}

\section{Introduction}

In the problem of minimal perfect hashing, we are given a universe of keys $[n] = \{1,\cdots,n\}$ and a subset $\mathcal{A} \subset [n]$ of size $|\mathcal{A}| = k$, and are tasked with constructing a hash function $h: [n] \to [k]$ which exhibits no collisions on $\mathcal{A}$, i.e., $h(\cdot)$ is a bijection when restricted to $\mathcal{A}$.
In this case, we say that $h(\cdot)$ is a minimal perfect hash function for $\mathcal{A}$, where "minimal" refers to the range of $h(\cdot)$ being the same cardinality as $\mathcal{A}$. 

Perfect hashing sees many applications throughout computer science, for example, in the design of hash tables, where we wish to map a large universe of keys to a smaller set of values representing positions on a lookup table  \cite{ph-survey}. When a collision occurs, costly followup measures are typically required which degrade performance. In many cases however, the set of keys that we wish to hash are known a priori. In such cases, we can employ a minimal perfect hashing scheme to construct a hash function a priori for the given set of keys, thus avoiding collisions during future use. 

More recently, perfect hashing has also found applications in communications, particularly in massive random access settings, where a central base-station wishes to provide connectivity to a small random subset of active users among a large pool of potential users. Instead of relying on a traditional contention based scheme such as ALOHA \cite{alohanet}, which requires frequent collision resolution, it has been proposed in \cite{scheduling,scheduling-vs-contention} that the base-station can instead first detect the identities the active users, then construct a perfect hash function which maps each active user to a unique transmission slot. Then, by communicating this perfect hashing function to the users prior to their transmissions, collisions between users can be avoided entirely. These ideas were then generalized to tackle the problem of communicating arbitrary information sources in a massive random access setting \cite{coded-downlink,lossy-coded-downlink}.

When evaluating the performance of a perfect hashing scheme, an important consideration is the number of bits required to store/communicate the hash function. Naively, any hash function $h: [n] \to [k]$ can be stored using $n\log(k)$ bits, which for large $n$, is infeasible. To design a more efficient scheme, we make two observations. Firstly, we only care about the output of the hash function $h(\cdot)$ over $\mathcal{A}$, the $k$ keys of interest. Secondly, although we require that $h(\cdot)$ be a bijection over $\mathcal{A}$, we do not necessarily need to describe a particular bijection among the $k!$ possibilities. Leveraging these two observations, it has been shown that the theoretically optimal amount of space required to store a perfect hash function is approximately $k\log(e)$ bits, or $\log(e)$ bits per key \cite{ph-1,ph-2,ph-3,scheduling}. Although deterministic schemes are theoretically capable of achieving the optimal $\log(e)$ bits per key, most modern practical perfect hashing schemes rely on randomization in the construction process \cite{chd, shockhash}. 

In this paper, we ask the following question. If we relax the collision-free requirement to instead only requiring that an $\alpha$-proportion of keys be hashed without collision in expectation, how much space would be required to store the hash function? This problem, which we call minimal $\alpha$-perfect hashing, extends the minimal perfect hashing to the approximate setting. When $\alpha = 1$, we recover the problem of minimal perfect hashing. The goal of this paper is to develop schemes for minimal $\alpha$-perfect hashing and to analyze their space requirements, which establishes achievability bounds on the space required for $\alpha$-perfect hashing. 


One may be tempted to approach $\alpha$-perfect hashing in the following way. Given a size $k$ subset $\mathcal{A}$, first select an $\alpha k$ sized subset of $\mathcal{A}$, then construct a perfect hash function for this subset, requiring $\alpha k\log(e)$ bits. The problem with this approach is that although the $\alpha k$ keys are guaranteed to be collision-free among themselves, they may end up colliding with the other $(1-\alpha)k$ keys in $\mathcal{A}$. Further, it turns out that one can develop $\alpha$-perfect hashing schemes which require much fewer bits than $\alpha k\log(e)$. This is because by considering a specific $\alpha k$ subset of $\mathcal{A}$, we implicitly store unnecessary information regarding that subset. Although we wish to control the total expected number of collisions, we are indifferent towards which particular keys end up in collisions.

\subsection{Main Results}

This paper proposes and analyzes the space requirements of two different $\alpha$-perfect hashing schemes. The first scheme randomizes between minimal perfect hashing and zero-bit hashing to achieve a space requirement of $\lambda(\alpha)\log(e)$ bits per key in expectation, where
\begin{equation}
    \lambda(\alpha) = \max \left\{ \frac{\alpha - 1/e}{1 - 1/e},0 \right\}.
\end{equation}
This essentially achieves a linear tradeoff between the number of collisions and the amount of space required. 

We then show that this linear relationship is not optimal. Using results from sampling theory \cite{sfrl,harsha}, we develop a sampling-based minimal $\alpha$-perfect hashing scheme which achieves a space requirement of
\begin{equation}
\begin{split}
 \hspace{-0.6em}\lambda(\alpha) \left( \log(e) - \log\left( \frac{1-\frac{\lambda(\alpha)}{2}}{1 -\frac{\lambda(\alpha)}{2} - \left( 1 -\frac{1}{e} \right) (1-\lambda(\alpha)) }\right) \right)
\end{split}
\end{equation}
bits per key in expectation. This improves upon the aforementioned scheme for all $\alpha$.

From Figure \ref{fig:ph-plot}, we can see a surprising tradeoff between $\alpha$ and the bits per key needed to store the hash function. For example, storing a minimal $0.9$-perfect hash function requires roughly $1$ bit per key compared to the $\log(e) \approx 1.44$ bits per key required to store a minimal perfect hash function. Note that when $\alpha = 1$, both hashing schemes achieve the optimal $\log(e)$ bits per key. 

\begin{figure}
    \centering
    \includegraphics[width=0.48\textwidth]{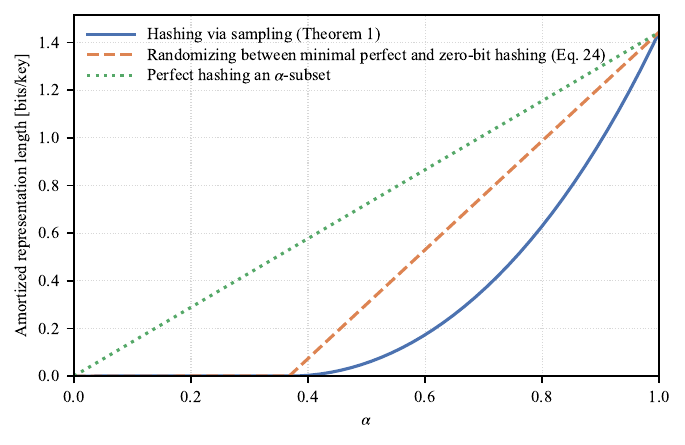}
    \caption{\label{fig:ph-plot}{\color{black} Upper bounds on the optimal amortized representation length $R^*(\alpha)$.}} 
\end{figure}

\subsection{Notation}

Unless stated otherwise, lowercase letters denote scalars, lowercase boldface letters denote vectors, capital letters denote random variables, boldface capital letters denote random vectors, and calligraphic letters, e.g. $\mathcal{S}$, denote sets.
All information measures are expressed in bits, including entropy $H(\cdot)$ and the Kullback–Leibler divergence $D(\cdot\|\cdot)$. We let $\log(\cdot)$ denote the base $2$ logarithm and use the shorthand notations $(x)^+ = \max\{ x,0 \}$, $[n] = \{ 1,\cdots,n \}$, and $\binom{[n]}{k}$ to denote the set of all size $k$ subsets of $[n]$.

\section{Problem Formulation}

This paper considers the problem of minimal $\alpha$-perfect hashing. Let $\mathcal{H}_{n,k} = [k]^n$ denote the set of all functions mapping $[n]$ to $[k]$ and $\alpha \in [0,1]$. Given a size $k$ subset $\mathcal{A} \in \binom{[n]}{k}$, we say that $a \in \mathcal{A}$ collides under $h(\cdot)$ in $\mathcal{A}$ if there exists another $a' \in \mathcal{A}$ such that $h(a) = h(a')$. We define $d: \binom{[n]}{k} \times \mathcal{H}_{n,k} \to \mathbb{R}$ to be the function which computes the proportion of $\mathcal{A}$ that collides under $h(\cdot)$, i.e.,
\begin{equation}
    d(\mathcal{A}, h) = \frac{1}{k}\Big|\big\{ a \in \mathcal{A} : \exists \ a' \in \mathcal{A}, \ a' \neq a, \ h(a) = h(a') \big\}\Big|.
\end{equation}
A hash function $h(\cdot)$ is said to be $\alpha$-perfect for $\mathcal{A}$ if $h(\cdot)$ exhibits at most $(1-\alpha)k$ collisions when restricted to $\mathcal{A}$, i.e.,
\begin{equation}
    d(\mathcal{A},h) \leq 1 - \alpha.
\end{equation}

A randomized hashing scheme takes as input a set $\mathcal{A}$ and stochastically constructs a hash function $h(\cdot)$. If the constructed hash function is $\alpha$-perfect in expectation, we say that this scheme is an $\alpha$-perfect hashing scheme. The goal of this paper is to design and analyze the space requirements of such hashing schemes. To do so, we define a randomized hashing scheme to consist of an encoder, a decoder, and common randomness $Z \in \mathcal{Z}$ which is accessible by both the encoder and decoder. The encoder
\begin{equation}
\label{eq:encoder}
    f: \binom{[n]}{k} \times \mathcal{Z} \to \{0,1\}^*
\end{equation}
maps a set $\mathcal{A}$ and a realization of the common randomness $z$ to a variable-length binary description. 
To ensure that the decoder knows when to stop reading the binary input,
we require that for each fixed $z$, the set of all possible strings output
by the encoder $f$ forms a prefix-free collection.
Using the encoder output along with the same $z$, the decoder 
\begin{equation}
    g: \{0,1\}^* \times \mathcal{Z} \to \mathcal{H}_{n,k}    
\end{equation}
recovers a hash function $h(\cdot)$. We say that a hashing scheme $(f,g,Z)$ is $\alpha$-perfect if for all inputs $\mathcal{A} \in \binom{[n]}{k}$, the produced hash function is $\alpha$-perfect for $\mathcal{A}$ in expectation, i.e.,
\begin{equation}
\label{eq:distortion-requirement}
    \max_{\mathcal{A} \in \binom{[n]}{k}} \mathbb{E}[d(\mathcal{A}, g(f(\mathcal{A}, Z), Z))] \leq 1-\alpha.
\end{equation}
It is worth noting that the distribution of $Z$ is designed a priori, without knowledge of $\mathcal{A}$. 

The encoder output is a binary representation of the hash function. Therefore, we define the space required by a hashing scheme $(f,g,Z)$ to be the maximum expected representation length over all inputs:
\begin{equation}
    L = \max_{\mathcal{A} \in \binom{[n]}{k}} \mathbb{E}[\mathtt{len}(f(\mathcal{A},Z))].
\end{equation}
We then define the optimal expected representation length $L^*(n,k,\alpha)$ to be the infimum expected representation length over all minimal $\alpha$-perfect hashing schemes with parameters $n$ and $k$. Another quantity of interest is the number of bits per key required as $n$ and $k$ become large. We define the optimal amortized expected representation length to be
\begin{equation}
    R^*(\alpha) =  \limsup_{k \to \infty} \limsup_{n \to \infty} \frac{1}{k} L^*(n,k,\alpha).
\end{equation}


\section{Minimal Perfect and Zero-Bit Hashing}
\label{sec:perf-zero}

In this section, we present a strategy for minimal $\alpha$-perfect hashing which consists of randomizing between minimal perfect hashing and zero-bit hashing. This strategy relies on the following observation. Given an $\alpha_1$-perfect hashing scheme $(f_1, g_1, Z_1)$ and an $\alpha_2$-perfect hashing scheme $(f_2, g_2, Z_2)$ which have expected representation lengths $L_1$ and $L_2$, respectively. Then, for any $\lambda \in [0,1]$, we can randomize between the two schemes by using the first scheme with probability $\lambda$ and the second scheme with probability $1 - \lambda$. The resulting new hashing scheme is $(\lambda\alpha_1 + (1-\lambda)\alpha_2)$-perfect and has expected representation length $\lambda L_1 + (1-\lambda)L_2$. As a consequence of this observation, we have that for any fixed $n$ and $k$, the function $L^*(n,k,\alpha)$ is convex in $\alpha$. 

\subsection{Minimal Perfect Hashing}

We begin by considering the problem of minimal $1$-perfect hashing where we ensure no collisions. From \cite{scheduling}, we know that the expected representation length
\begin{equation}
    L^*(n,k,1) \leq k\log(e) + 3
\end{equation}
is achievable. Therefore, the amortized expected representation length is bounded from above as
\begin{equation}
    R^*(1) \leq \log(e).
\end{equation}
We recap the proof of these bounds as the ideas and techniques employed will become useful in later sections.

The main argument hinges on the use of uniformly random hash functions which are identically and independently distributed (i.i.d.). Let 
\begin{equation}
    Z = \left( C^{(1)}, C^{(2)}, \cdots \right)
\end{equation}
be a sequence of i.i.d.\ random hash functions available at both the encoder and decoder, where each $C^{(t)} \sim \mathrm{Uniform}(\mathcal{H}_{n,k})$.
Given a set $\mathcal{A} \in \binom{[n]}{k}$ and a realization $z = \left( c^{(1)}, c^{(2)}, \cdots \right)$, one can 
search through $z$ to find the first index $t$ such that $d\left(\mathcal{A}, c^{(t)} \right) = 0$, and
then provide a binary description of $t$, that is, we define
\begin{equation}
    l(\mathcal{A}, z) = \min \left\{ t > 0  :  d\left(\mathcal{A}, c^{(t)} \right) = 0 \right\}
\end{equation}
and let $f(\mathcal{A},z)$ be the binary description of the positive integer $l(\mathcal{A}, z)$ via a prefix-free source code. We may design this source code optimally based on the probability distribution of $l(\mathcal{A}, Z)$. The decoder $g:\{0,1\}^* \times \mathcal{Z} \to \mathcal{H}_{n,k}$, given the
binary description and $z \in\mathcal{Z}$, can recover $l(\mathcal{A},z)$ first then use $z$ to produce the hash function $c^{(l(\mathcal{A},z))}$.

Since every $C^{(t)}$ is i.i.d.\ uniform over $\mathcal{H}_{n,k}$, for every $\mathcal{A} \in \binom{[n]}{k}$, the index $l(\mathcal{A}, Z)$ is geometrically distributed as
\begin{equation}
    \mathrm{Pr}(l(\mathcal{A},Z) = t) = \left( 1 - \frac{k!}{k^k} \right)^{t-1} \frac{k!}{k^k}.
\end{equation}
As the entropy $H(l(\mathcal{A},Z))$ is upper bounded by
\begin{align}
    H(l(\mathcal{A},Z)) &\leq \log\left( \frac{k^k}{k!} \right) + \log(e) \\
    &\leq k\log(e) + 2,
\end{align}
we can find a binary prefix-fix free code $\mathtt{c}$ for the positive integers with $\mathbb{E}[\mathtt{len}(\mathtt{c}(l(\mathcal{A},Z)))] \leq k\log(e) + 3$. Therefore 
\begin{align}
    L^*(n,k,1) &\leq H(l(\mathcal{A},Z)) + 1 \\
    &\leq k\log(e) + 3.
\end{align}
The upper bound on $R^*(1)$ follows immediately.



\subsection{Minimal Zero-Bit Hashing}

We now consider the conceptual opposite of perfect hashing, where regardless of the input, the encoder outputs the empty string and the decoder outputs a hash function selected uniformly at random from $\mathcal{H}_{n,k}$. Since the representation length is $0$ bits, we refer to this strategy as zero-bit hashing. 

Fix $n$, $k$, and a set $\mathcal{A} \in \binom{[n]}{k}$. Let $C \sim \mathrm{Uniform}(\mathcal{H}_{n,k})$ be a uniformly random hash function and define random variables $(B_1, \cdots, B_k) \in \{0,1\}^k$ to be
\begin{equation}
    B_i = \begin{cases}
        1, \quad \text{if } \  \exists! \ j \in \mathcal{A} \text{ such that } C(j) = i  \\
        0, \quad \text{otherwise,}
    \end{cases}
\end{equation}
where $B_i$ is the indicator random variable for the event that exactly one element of $\mathcal{A}$ hashes to value $i$. The number of elements in $\mathcal{A}$ which are involved in a collision is exactly
\begin{equation}
    D = k -\sum_{i=1}^{k} B_i.
\end{equation}
Observing that
\begin{align}
    \mathbb{E}[B_i] &= \binom{k}{1} \frac{1}{k}\left(1-\frac{1}{k}\right)^{k-1} \geq \frac{1}{e}
\end{align}
and applying the linearity of expectation, we have that
\begin{equation}
    \mathbb{E}[D] \leq k \left( 1 - \frac{1}{e} \right).
\end{equation}
Since in expectation, at least a $\frac{1}{e}$-proportion of $\mathcal{A}$ are mapped without collision, the zero-bit hashing scheme is a $\frac{1}{e}$-perfect hashing scheme. Therefore, for $0 \leq \alpha \leq \frac{1}{e}$, we have that $L^*(n,k,\alpha) = 0$
and $R^*(\alpha) = 0$.

\subsection{Randomizing Between Perfect and Zero-Bit Hashing}

For values of $\alpha$ between $\frac{1}{e}$ and $1$, we can randomize between the aforementioned minimal perfect and zero-bit hashing schemes by using the minimal perfect hashing scheme with some probability $\lambda \in [0,1]$ and the zero-bit hashing scheme otherwise. Doing so yields a hashing scheme that is $\left(\lambda + (1-\lambda)\frac{1}{e}\right)$-perfect and achieves an expected representation length of $\lambda(k\log(e) + 3)$ bits. Rearranging and finding the relationship between $\lambda$ and $\alpha$, we can see that the randomized scheme is $\alpha$-perfect as long as $\lambda \geq \frac{\alpha - \frac{1}{e}}{1 - \frac{1}{e}}$. This implies that, 
\begin{equation}
    L^*(n,k,\alpha) \leq \left( \frac{\alpha -1/e }{1 - 1/e} (k\log(e) + 3) \right)^+
\end{equation}
and
\begin{equation}
\label{eq:base-line}
    R^*(\alpha) \leq \left( \frac{\alpha -1/e}{1 - 1/e} \right)^+ \log(e) .
\end{equation}
In the next section, we present a hashing scheme based on sampling which improves upon the above strategy for all $\alpha$.

\section{Hashing via Sampling}

In this section, we review techniques from sampling and apply them to develop an improved $\alpha$-perfect hashing scheme. Let $M = \left(X^{(1)}, X^{(2)}, \cdots\right)$ be an infinite sequence of random variables which are i.i.d.\ according to distribution $q(x)$ and take values on a countable set $\mathcal{X}$. In the problem of sampling, we are given a realization of the sequence $\left(x^{(1)}, x^{(2)}, \cdots \right)$ from which we wish to select a single sample $x^{(t)}$ such that $x^{(t)}$ is distributed according to distribution $p(x)$. To incorporate randomized strategies, we let $U \in \mathcal{U}$ be a random variable which is independent of $M$ which is to be designed a priori. Formally, the problem of sampling is to find a function
\begin{equation}
    s: \mathcal{X}^{*} \times \mathcal{U} \to \mathbb{N}
\end{equation}
such that the $s(M,U)$'th term in $M$ is distributed according to $p(x)$, i.e.,
\begin{equation}
    X^{(s(M,U))} \sim p(x).
\end{equation}
Beyond finding an $s(\cdot)$ which generates the correct target distribution $p(x)$, it is also of interest to find an $s(\cdot)$ whose entropy $H(s(M,U))$ is small.

To this end, we apply the results of \cite{sfrl}, which takes advantage of the properties of the Poisson point process to perform sampling. In \cite[Theorem 1]{sfrl}, it is shown that there exists a $s(\cdot)$ and $U$ such that
\begin{equation}
    X^{(s(M,U))} \sim p(x)
\end{equation}
which also has entropy bounded from above as
\begin{equation}
    H(s(M,U)) \leq D(p\|q) + \log(D(p\|q) + 1) + 4.
\end{equation}
We refer this sampling scheme as the Poisson Function Representation (PFR) sampling scheme.

We can now apply the above results on sampling towards $\alpha$-perfect hashing. Similar to how we designed scheme for minimal perfect hashing, let
\begin{equation}
    Z = \left( C^{(1)}, C^{(2)}, \cdots \right)
\end{equation}
be an infinite sequence of i.i.d.\ random hash functions, where each $C^{(t)} \sim \mathrm{Uniform}(\mathcal{H}_{n,k})$. Note that for any set $\mathcal{A} = \{a_1, \cdots, a_k\} \in \binom{[n]}{k}$, the restriction of $C^{(t)}$ to $\mathcal{A}$
\begin{equation}
    \mathbf{X}^{(t)} = \left( C^{(t)}(a_1), \cdots, C^{(t)}(a_k) \right) 
\end{equation}
are i.i.d.\ uniformly over $[k]^k$ for all positive integers $t > 0$. Therefore, given a distribution $p(\mathbf{x})$ on $[k]^k$, the PFR sampling scheme provides us with an index selection method $s(\mathcal{A},Z,U)$ which ensures that for all $\mathcal{A}$, we have that
\begin{equation}
    \mathbf{X}^{(s(\mathcal{A},Z,U))} \sim p(\mathbf{x}).
\end{equation}
Further, the entropy is bounded from above as
\begin{equation}
    H(s(\mathcal{A},Z,U)) \leq D(p\|q) + \log(D(p\|q) + 1) + 4.
\end{equation}
Since the hash functions are distributed uniformly over $\mathcal{H}_{n,k}$, the distribution of $s(\mathcal{A},Z,U)$ is the same for any choice of $\mathcal{A}$. Therefore, there exists a binary prefix-free code $\texttt{c}$ for the positive integers such that
\begin{equation}
    \mathbb{E}[\texttt{len}(\texttt{c}(s(\mathcal{A},Z,U)))] \leq H(s(\mathcal{A},Z,U)) + 1,
\end{equation}
regardless of the choice of $\mathcal{A}$. 

With this, we define the following hashing scheme. Let $(Z,U)$ be common randomness. Given input $\mathcal{A}$, the encoder computes $s(\mathcal{A},z,u)$ then compresses $s(\mathcal{A},z,u)$ using binary prefix-free code $\texttt{c}$.
Given this binary description, the decoder first recovers $s(\mathcal{A},z,u)$ then the hash function $c^{(s(\mathcal{A},z,u))}$.

The expected number of collisions for this hashing scheme depends only on the choice of $p(\mathbf{x})$.
To ensure that the scheme is $\alpha$-perfect, we should choose $p(\mathbf{x})$ so that 
\begin{equation}
\label{eq:perf-req-x}
    \mathbb{E}[d(\mathbf{X})] \leq 1-\alpha,
\end{equation}
where $\mathbf{X} \sim p(\mathbf{x})$ and
\begin{equation}
    d(\mathbf{x}) = \frac{1}{k} \Big|\big\{i \in [k] : \exists j \in [k], j \neq i, x_j = x_i \big\}\Big|.
\end{equation}

The expected representation length on the other hand scales with the divergence $D(p\|q)$. Since the distribution $q(\mathbf{x})$ is the uniform distribution over $[k]^k$, we have that
\begin{equation}
    D(p\|q) = k\log(k) - H(\mathbf{X}),
\end{equation}
where $\mathbf{X} \sim p(\mathbf{x})$. Thus, for $\alpha$-perfect hashing, we wish to design a distribution $p(\mathbf{x})$ which maximizes the entropy subject to the constraint $\mathbb{E}[d(\mathbf{X})] \leq 1 - \alpha$. The following lemma restates the above reasoning in terms of $L^*(n,k,\alpha)$.

\begin{lem}
\label{lem:frl-coding}
    Let $\mathbf{X} = (X_1, \cdots, X_k)$ be a sequence of random variables each taking values on $[k]$ such that
    \begin{equation}
        \mathbb{E}[d(\mathbf{X})] \leq 1-\alpha.   
    \end{equation}
    Then, the optimal representation length $L^*(n,k,\alpha)$ for minimal $\alpha$-perfect hashing is bounded from above as
    \begin{equation}
    \begin{split}
        L^*(n,k,\alpha) &\leq k\log(k) - H(\mathbf{X}) \\ 
                &\qquad\quad + \log(k\log(k) - H(\mathbf{X}) + 1) + 5.
    \end{split}
    \end{equation}
\end{lem}

\section{Space Bounds for Minimal $\alpha$-Perfect Hashing}

In this section, we develop upper bounds on $L^*(n,k,\alpha)$ and $R^*(\alpha)$ by first constructing a suitable sequence of random variables $\mathbf{X} = (X_1, \cdots, X_k)$ which satisfy the collision constraint $\mathbb{E}[d(\mathbf{X})] \leq 1-\alpha$ while simultaneously having large entropy, then applying Lemma \ref{lem:frl-coding}. Let $\mathbf{X} \in [k]^k$ be defined in the following way:
\begin{enumerate}
    \item Fix a value $\lambda \in [0,1]$. 
    \item Let $w = \lceil \lambda k \rceil$.
    \item Construct random variables $(Y_1, \cdots, Y_k) \in [k]^k$ as follows. Let $(Y_1, \cdots, Y_w)$ be the sequence resulting from $w$ draws \emph{without replacement} from an urn containing elements $[k]$. Next, let $(Y_{w+1}, \cdots, Y_k)$ be the sequence resulting from $k-w$ additional draws \emph{with replacement} from the same urn.
    \item Let $\pi(\cdot)$ denote a uniformly random permutation of $[k]$ and let $\mathbf{X}$ be a random permutation of $\mathbf{Y}$, i.e.,
    \begin{equation}
        \mathbf{X} = (X_1, \cdots, X_k) = \left(Y_{\pi(1)}, \cdots, Y_{\pi(k)}\right).
    \end{equation}
\end{enumerate}
For some $\mathbf{X}$ generated with parameter $\lambda$, we denote the distribution of $\mathbf{X}$ as $r_\lambda(\mathbf{x})$. 

\begin{lem}
\label{lem:distortion-x}
    Let $\lambda \in [0,1]$ and $\mathbf{X} = (X_1, \cdots, X_k)$ be distributed according to distribution $r_\lambda(\mathbf{x})$. Then,
    \begin{equation}
        \mathbb{E}[d(\mathbf{X})] \leq \left( 1 - \frac{1}{e} \right)(1-\lambda).
    \end{equation}
\end{lem}
\begin{IEEEproof}
    By the definition of $\mathbf{X}$, exactly $\lceil \lambda k \rceil$ entries of $\mathbf{X}$ are associated with a draw without replacement, and therefore are guaranteed to not be involved in any collision. The remaining $k-\lceil \lambda k \rceil$ entries of $\mathbf{X}$ are mapped randomly to the remaining $k-\lceil \lambda k \rceil$ elements in the urn. Therefore, in expectation, the expected number of collisions is bounded from above as $\left( 1-\frac{1}{e}\right)(k- \lceil k \lambda \rceil) \leq k\left( 1-\frac{1}{e}\right)(1 - \lambda).$ 
\end{IEEEproof}

\begin{lem}
\label{lem:entropy-x}
    Let $\lambda \in [0,1]$ and $\mathbf{X} = (X_1, \cdots, X_k)$ be distributed according to the distribution $r_\lambda(\mathbf{x})$. Then,
    \begin{equation}
    \begin{split}
        H(&\mathbf{X}) \geq k\log(k) - \lambda k\log(e) \\
        &+ \lambda k\log\left( \frac{1-\frac{\lambda - 1/k}{2}}{1 -\frac{\lambda - 1/k}{2} - \left(1 - \frac{1}{e} \right)(1-\lambda - 1/k)  }\right).
    \end{split}
    \end{equation}
\end{lem}
\begin{IEEEproof}
    To lower bound $H(\mathbf{X})$, we use the intermediate random variables $\mathbf{Y}$, which are the values of $\mathbf{X}$ prior to being randomly permuted. The idea is to first bound $H(\mathbf{Y})$, then bound the difference $H(\mathbf{X}) - H(\mathbf{Y})$. For $H(\mathbf{Y})$, we can use Stirling's approximation to show that
    \begin{align}
    \label{eq:lb-on-y}
        \hspace{-1em}H(\mathbf{Y}) &= \log\left( \binom{k}{w} w! \right) + (k-w)\log(k-w) \\
        &\geq k\log(k) - w\log(e),
    \end{align}
    where $w = \lceil \lambda k \rceil$.
    
    Next, we wish to lower bound the difference $H(\mathbf{X}) - H(\mathbf{Y})$, which is the entropy increase attributed to randomly permuting $\mathbf{Y}$. Let $v : [k]^k \to \{0,1\}^{k}$ be the collision indicator function, i.e., given $\mathbf{x}\in [k]^k$, $v(\mathbf{x})_i = 1$ if and only if there exists a $j \neq i$ such that $x_i = x_j$. Notice that
    \begin{align}
        &\hspace{-0.5em}H(\mathbf{X}) - H(\mathbf{Y}) = H(\mathbf{X}, v(\mathbf{X})) - H(\mathbf{Y}, v(\mathbf{Y})) \\
        &= H(\mathbf{X} | v(\mathbf{X})) - H(\mathbf{Y} | v(\mathbf{Y})) + H(v(\mathbf{X})) - H(v(\mathbf{Y})). 
    \end{align}

    We claim that $H(\mathbf{X} | v(\mathbf{X})) = H(\mathbf{Y} | v(\mathbf{Y}))$. We omit a formal proof and provide some intuition instead. Observe that the conditional probability $\mathrm{Pr}(\mathbf{X}=\mathbf{x}|v(\mathbf{X}) = \mathbf{v})$ depends only on the weight of $\mathbf{v}$, and not the ordering. This, along with Bayes' rule, can be used to show that $\mathrm{Pr}(\mathbf{X}=\mathbf{x}|v(\mathbf{X}) = \mathbf{v})$ is equal to $\mathrm{Pr}(\mathbf{Y} = \mathbf{y} | v(\mathbf{Y}) = \mathbf{v})$ for all $\mathbf{v}$ such that $\mathrm{Pr}(v(\mathbf{Y}) = \mathbf{v}) > 0$. This implies that the conditional entropies are equal and hence
    \begin{align}
        H(\mathbf{X}) - H(\mathbf{Y}) =  H(v(\mathbf{X})) - H(v(\mathbf{Y})). 
    \end{align}
  
    Let $D = \mathrm{wt}(v(\mathbf{X}))$, where $\mathrm{wt}(\cdot)$ denotes the weight function. Since $\mathbf{X}$ is a permutation of $\mathbf{Y}$, $D = \mathrm{wt}(v(\mathbf{X})) = \mathrm{wt}(v(\mathbf{Y}))$, which implies that  
    \begin{equation}
    \begin{split}
        H(v(\mathbf{X})) - H(v(\mathbf{Y})) = H(v(\mathbf{X}) | D) - H(v(\mathbf{Y}) | D).
    \end{split}
    \end{equation}
    Conditioned on $D$, the sequence $\mathbf{X}$ is uniformly distributed over binary strings of weight $D$, and $\mathbf{Y}$ is uniformly distributed over strings of weight $D$ which begin with $w$ zeros. Using this fact and along with multiple applications of Jensen's inequality, we get
    \begin{align}
         H&(v(\mathbf{X}) | D) - H(v(\mathbf{Y}) | D) \nonumber \\
         &= \mathbb{E}\left[ \log \binom{k}{D} \right] -  \mathbb{E}\left[ \log\binom{k-w}{D} \right] \\
         &= \mathbb{E}\left[ \sum_{i=0}^{w-1} \log\left( \frac{k-i}{k-i-D} \right) \right] \\
         &\geq \mathbb{E}\left[ w \log\left( \frac{k-\frac{w-1}{2}}{k-\frac{w - 1}{2}-D} \right) \right] \\
         &\geq \lambda k \log\left( \frac{1-\frac{\lambda - 1/k}{2}}{1- \frac{\lambda - 1/k}{2} - \left(1 - \frac{1}{e} \right)(1 -\lambda - 1/k) } \right). \label{eq:diff-bound}
    \end{align}
    The last line follows from the fact that $w \geq \lambda k$ and $\mathbb{E}[D] \geq k \left( 1- \frac{1}{e}\right) \left(1 - \lambda  - 1/k \right)$.  Combining \eqref{eq:diff-bound} with \eqref{eq:lb-on-y} yields the desired result.
\end{IEEEproof}

Together with Lemmas \ref{lem:distortion-x} and \ref{lem:entropy-x}, Lemma \ref{lem:frl-coding} immediately implies the following upper bounds on $L^*(n,k,\alpha)$ and $R^*(\alpha)$. 


\begin{thm}
\label{thm:main}
    Consider the problem of minimal $\alpha$-perfect hashing. For fixed $n$ and $k$, the optimal representation length and optimal amortized representation length for $\alpha$-perfect hashing are bounded from above as
    \begin{equation}
    \begin{split}
        &L^*(n,k,\alpha) \leq \lambda(\alpha) k \log(e) \\
        &- \lambda(\alpha)k\log\left( \frac{1-\frac{\lambda(\alpha) - 1/k}{2}}{1 -\frac{\lambda(\alpha)-1/k}{2} - \left( 1-\frac{1}{e} \right) (1-\lambda(\alpha) - 1/k) }\right) \\
        &+ O(\log k)
    \end{split}
    \end{equation}
    and
    \begin{equation}
    \begin{split}
        R&^*(\alpha) \leq  \lambda(\alpha) \log(e) \\
        &\quad - \lambda(\alpha)\log\left( \frac{1-\frac{\lambda(\alpha)}{2}}{1 -\frac{\lambda(\alpha)}{2} - \left( 1 -\frac{1}{e} \right) (1-\lambda(\alpha)) }\right),
    \end{split}
    \end{equation}
    respectively, where
    \begin{equation}
        \lambda(\alpha) = \left( \frac{\alpha - 1/e}{1 - 1/e} \right)^+.
    \end{equation}
\end{thm}

For a comparison between the results of Theorem~\ref{thm:main} and the bound \eqref{eq:base-line} achieved through randomizing between minimal perfect and zero-bit hashing, refer to Figure \ref{fig:ph-plot}. 

\section{Conclusion}
This paper introduces the problem of $\alpha$-perfect hashing, an extension of minimal perfect hashing to the approximate setting, where we require only that an $\alpha$-proportion of keys be hashed without collision rather than all. We present two different $\alpha$-perfect hashing schemes and analyze their respective space requirements. The first scheme randomizes between minimal perfect and zero-bit hashing. The second leverages results from sampling theory to develop an $\alpha$-perfect hashing scheme which significantly improves upon the space requirement of the aforementioned scheme for all $\alpha$.

\bibliographystyle{IEEEtran}
\balance
\bibliography{references}

\end{document}